\documentclass[prl,twocolumn,showpacs]{revtex4-2}

\usepackage{amsmath}
\usepackage{latexsym}
\usepackage{amssymb}
\usepackage{graphicx}
\usepackage[colorlinks=true, citecolor=blue, urlcolor=blue]{hyperref}
\usepackage{float}
\usepackage{amsfonts}
\usepackage{textcomp}
\usepackage{mathpazo}
\usepackage{comment}
\usepackage{xr}

\usepackage{bbm}

\usepackage{xcolor}
\definecolor{myurlcolor}{rgb}{0,.5,.5}
\definecolor{mycitecolor}{rgb}{0,.6,0}
\definecolor{myrefcolor}{rgb}{2,0,0}
\usepackage{hyperref}
\hypersetup{colorlinks,
linkcolor=myrefcolor,
citecolor=mycitecolor,
urlcolor=myurlcolor}

\usepackage[draft]{fixme}
\usepackage{amsmath,bbm}
\usepackage{graphicx}
\usepackage{amsfonts}
\usepackage{amssymb}
\usepackage{amsmath, amssymb, amsthm,verbatim,graphicx,bbm}
\usepackage{mathrsfs}
\usepackage{color,xcolor,longtable}

\usepackage{xr}
\makeatletter
\newcommand*{\addFileDependency}[1]{
  \typeout{(#1)}
  \@addtofilelist{#1}
  \IfFileExists{#1}{}{\typeout{No file #1.}}
}
\makeatother

\newcommand*{\myexternaldocument}[1]{
    \externaldocument{#1}
    \addFileDependency{#1.tex}
    \addFileDependency{#1.aux}
}

\myexternaldocument{manuscript_PRL}

\newcommand{\beq}[0]{\begin{equation}}
\newcommand{\eeq}[0]{\end{equation}}

\newcommand{\one}{\leavevmode\hbox{\small1\normalsize\kern-.33em1}}

\def\be{\begin{equation}}
\def\ee{\end{equation}}
\def\ben{\begin{eqnarray}}
\def\een{\end{eqnarray}}
\def\eea{\end{array}}
\def\bea{\begin{array}}

\newcommand{\Tr}[1]{\mathrm{Tr}#1}
\newcommand{\bei}{\begin{itemize}}
\newcommand{\eei}{\end{itemize}}
\newcommand{\ket}[1]{|#1\rangle}
\newcommand{\bra}[1]{\langle#1|}

\newcommand{\proj}[1]{\ket{#1}\!\!\bra{#1}}

\newcommand{\I}{\mathbbm{1}}

\renewcommand{\emph}[1]{\textbf{#1}}

\makeatletter
\newtheorem*{rep@theorem}{\rep@title}
\newcommand{\newreptheorem}[2]{%
\newenvironment{rep#1}[1]{%
 \def\rep@title{#2 \ref{##1}}%
 \begin{rep@theorem}}%
 {\end{rep@theorem}}}
\makeatother

\theoremstyle{plain}
\newtheorem{thm}{Theorem}
\newtheorem*{thm*}{Theorem}
\newreptheorem{thm}{Theorem}

\theoremstyle{definition}

\theoremstyle{remark}

\usepackage[T1]{fontenc}

\begin{document}

\title{Entanglement is not sufficient for most practical entanglement-based QKD protocols}

\author{
Shubhayan Sarkar\textsuperscript{1}, Tushita Prasad\textsuperscript{2}, and Karol Horodecki\textsuperscript{1} \\
\textit{\textsuperscript{1}Institute of Informatics, Faculty of Mathematics, Physics and Informatics,\\
University of Gdansk, Wita Stwosza 57, 80-308 Gdansk, Poland\\
\textsuperscript{2}ICTQT, University of Gdansk, Wita Stwosza 57, 80-308 Gdansk, Poland
}}

\begin{abstract}
Quantum key distribution (QKD) is the most explored application of quantum information theory. A central problem in entanglement-based QKD (EB-QKD), is whether every entangled state can be used to extract a key. We observe that entanglement is not sufficient for standard practical EB-QKD protocols where the input choices are announced by the parties that want to share a secure key, such as E91 or entanglement-based BB84 type protocols, when even an arbitrarily small amount of leakage of classical side information occurs. We do this by identifying a class of two-qubit isotropic states that are entangled but cannot be used to distil the key under such protocols for any possible measurement by the parties. Counter-intuitively, this gap persists even when the leakage occurs from the "junk" rounds of the protocol, i.e, rounds that cannot be used to generate any key. We then extend this result to arbitrary dimensions and parties by identifying a class of isotropic states that are not useful to extract a secure key under such protocols, even if they are entangled. Finally, we demonstrate that our approach provides a tool to upper-bound the scalability of repeater-based QKD architectures in a protocol-independent manner. Interestingly, we find that allowing for even a tiny noise in the preparation drastically reduces the scalability of the QKD network.
\end{abstract}


\maketitle

{\it{Introduction---}}Quantum Key Distribution (QKD) is one of the most prominent practical applications of quantum information science, enabling information-theoretically secure communication between spatially separated users based on the laws of quantum mechanics. The first QKD protocol, BB84 \cite{bennett1984quantum}, relies on encoding information in non-orthogonal quantum states, ensuring that any eavesdropping attempt necessarily introduces detectable disturbances. This paradigm was later extended to entanglement-based schemes, such as E91 \cite{Ekert1991E91}, which relates security to the violation of Bell inequalities, and BBM92 \cite{bennett1992quantum}, the entanglement-based version of BB84, where key generation is achieved through measurements on shared entangled states.

Security proofs of essentially all standard device-dependent QKD protocols including the six-state protocol \cite{Brus1998}, B92 \cite{Bennett1992}, decoy-state BB84 \cite{Tamaki2008}, and measurement-device-independent QKD \cite{Lo_2012} are based on an equivalent entanglement-based description. Indeed, it was shown in Ref.~\cite{Curty2004} that a necessary condition for secure key distribution is that the legitimate users can certify entanglement in the effectively distributed quantum state. Relatedly, protocols that rely on Bell inequality violations \cite{Primaatmaja_2023, Zapatero2023} necessarily require entangled sources. These results firmly establish entanglement as the fundamental resource underlying quantum key distribution.

This naturally raises a foundational question: Is the presence of entanglement sufficient to guarantee secure key distribution? In other words, can every entangled state be used to establish a secret key between spatially separated parties? Remarkably, Ref.~\cite{karol1} shows that even bound entangled states, from which pure entanglement cannot be distilled, can nevertheless yield a secure key. Moreover, the distillable key was later shown to be upper-bounded by the relative entropy of entanglement and refined in terms of the intrinsic information of an adversary \cite{intinf1}. Since these quantities are entanglement monotones \cite{RevModPhys.81.865}, one might be tempted to conclude that entanglement alone suffices for QKD.

In this work we show that this intuition breaks down once one moves beyond the idealized setting. 
 Practical QKD implementations are inevitably subject to imperfections and classical side-channel leakage. Experimental demonstrations and implementation-level attacks have shown how such leakage can compromise security \cite{leak1,leak2,leak3,leak4,leak5}. On the theoretical side, a substantial body of work has analysed QKD under specific leakage models arising from imperfect or flawed sources, including leaky modulators, Trojan-horse attacks, and model-dependent imperfections \cite{Tamaki2016, Pereira2019, Wang2021, Navarrete2022, Horodecki_2021}. 
Nevertheless, these works focus on particular leakage mechanisms and protocol-specific settings, leaving open a more fundamental question regarding the role of entanglement in the presence of general classical side information.

Addressing these aspects, we develop a general quantitative framework to analyze {\it{classical side-channel leakage}} in standard EB-QKD protocols and allows us to compute upper-bound to key rate in presence of leakage. 
Within this framework, we quantify a fundamental separation between entanglement and key extractability: even arbitrarily small leakage can create a gap between states that are entangled and states from which a secret key can actually be distilled. This gap persists even if the leakage occurs only from "junk" rounds of the protocol, that is, rounds from which no secure key could be extracted. 
By standard protocols, we mean those in which users directly extract the secret key from their measurement outcomes after publicly announcing their inputs, such as entanglement-based BB84 (BBM92), E91, and related schemes. These protocols may also be referred to as quantum-memoryless protocols, since the key is extracted via direct measurements on the received subsystems combined with public communication, without requiring quantum memory. Protocols that rely on additional iterative quantum procedures, such as the Gottesman–Lo protocol \cite{GottesmanLo2003}, fall outside this framework and are not considered here. 


To establish this separation between entanglement and key extractability explicitly, we apply our framework to entanglement-based protocols using isotropic states. We then explicitly construct an attack that renders the secret-key rate zero for a class of entangled states under arbitrarily small leakage. 
We extend the analysis to higher-dimensional and multipartite isotropic states and further derive numerical upper bounds on the maximal key rate achievable from any entangled state in the presence of leakage. We illustrate these bounds for isotropic states \cite{Dur2000}  in dimensions two and three, thereby obtaining quantitative, protocol-independent limits on key generation under classical side information.


Finally, to demonstrate the general applicability of our approach, we also consider repeater-based QKD architectures, which aim to overcome channel loss by dividing long communication distances into shorter elementary links connected via intermediate nodes. Using our approach, we derive fundamental bounds on the maximum number of repeaters that can be concatenated while still enabling secure key generation over standard protocols in the presence of classical side-channel leakage. These results provide operational limits on the scalability of practical QKD networks under leakage, highlighting how even a small amount of leaked information can drastically reduce the maximum achievable distance and number of repeaters, which is pivotal for the design of future quantum communication technologies.\\

{\it{Entanglement-based QKD}}---
Consider a scenario with $N$ parties, labeled $A_i$ for $i=1,\dots,N$, who aim to establish a shared secret key. In an EB-QKD protocol, each party receives a subsystem from a quantum source and performs measurements with freely chosen inputs $x_i$, producing outcomes $a_i$. For compactness, we denote the full set of inputs and outcomes as $\vec{x} = \{x_1,\dots,x_N\}$ and $\vec{a} = \{a_1,\dots,a_N\}$, respectively. The source prepares a joint state $\rho_{A_1\ldots A_N}$, and the measurement operators of party $A_i$ are $\{M_{a_i|x_i}\}$. The resulting correlations are captured by the joint probability distribution
\[
\vec{p} = \{ p(\vec{a}|\vec{x}) \}, \quad 
p(\vec{a}|\vec{x}) = \mathrm{Tr}\Bigg(\bigotimes_{i=1}^N M_{a_i|x_i} \, \rho_{A_1\ldots A_N}\Bigg).
\]

In this work, we focus on standard entanglement-based protocols in which the inputs of both parties are publicly announced. Such protocols are the most studied ones in QKD [for instance see Refs. \cite{Ekert1991E91, Bennett1992EntBB84, Tomamichel_2012, masini2024, Tomamichel_2017,  QKDrev1, QKDrev2}]. 

The execution of such protocols is typically divided into parameter estimation rounds and key generation rounds. In the parameter estimation stage, a randomly chosen subset of the raw data is publicly disclosed and used to estimate the relevant channel parameters, such as the quantum bit error rate (QBER) or correlations needed for security proofs. This step provides a statistical guarantee against potential eavesdropping and determines whether the protocol should be aborted. The undisclosed portion of the data constitutes the key generation rounds, from which the raw key is distilled. Subsequent post-processing steps, namely error correction and privacy amplification, are then applied to transform this raw key into a secure final secret key. 

The above description assumes that the adversary’s information is limited to what can be obtained through interaction with the quantum channel and the publicly announced communication. In practice, however, one can never negate the existence of classical side channels. For instance, when a detector produces a classical bit, that outcome is encoded in electrical signals such as current spikes, voltage changes, timing signals, or electromagnetic emissions. If these physical signatures differ even slightly depending on the measurement result, an eavesdropper may infer partial information about the outcome without directly interacting with the quantum channel . In security terms, such leakage creates additional correlations between the key and the adversary, effectively increasing Eve’s information. 
To model this effect quantitatively, we consider that each measurement outcome is leaked to the adversary with a finite probability $\mathcal{L}$.

{\it{Convex combination attacks---}}
Consider now that there is an adversary, Eve, who wants to intrude on the protocol and guess the key established between the parties. There exist a variety of attacks that allow Eve to remain undetected during the parameter-estimation rounds while still being able to access the key (see, e.g., Refs.~\cite{QKDrev1,QKDrev2} for a review).
Here, we consider the simplest yet one of the most realistic adversarial attacks, known as convex combination (C.C.) attack. Originally introduced in \cite{Acin_20061, Acin_2006} for the device-independent setting, the C.C.
strategy adapted here to the device-dependent setup involves reproducing correlations between Alice and Bob by mixing correlations generated by separable and entangled states, that is, given a correlation $\{p(\vec{a}|\vec{x})\}$ it can be expressed as,
\begin{equation}
    p(\vec{a}|\vec{x})=q_{\mathcal{S}} p^{\mathcal{S}}(\vec{a}|\vec{x})+(1-q_{\mathcal{S}}) p^{\mathcal{E}}(\vec{a}|\vec{x})
\end{equation}
where $\{p^{\mathcal{S}}(\vec{a}|\vec{x})\}$ can be realised using a separable state and $\{p^{\mathcal{E}}(\vec{a}|\vec{x})\}$ can be realised using an entangled state with $0\leq q_{\mathcal{S}}\leq1$. 

We introduce a classical variable $e\in \{\mathcal{S},\mathcal{E}\}$ 
that tells Eve which part of the mixture is being distributed in the current round, the separable state or the entangled state.
Using this labeling, we can now quantify the maximum key extractable from the protocol in terms of Eve’s knowledge. This is captured by the intrinsic information $I(A_1,A_2,\ldots,A_N\downarrow e)$, \cite{intinf1, intinf2, intinf3} where


\begin{equation}\label{intinf11}
   I(A_1\ldots A_N\downarrow e)_{\rho}=\inf_{\rho'=(\I_{A_1\ldots A_N}\otimes \Lambda_E )\rho} I(A_1\ldots A_N| e)_{\rho'}.
\end{equation}
Here, $\Lambda_E$ is any operation that can be performed by Eve with $I(A_1\ldots A_N| e)_{\rho'}$ denoting the mutual information of all parties conditioned on Eve. If the shared state is separable, then Eve can perfectly predict the outcomes associated with the separable component, thereby precluding the establishment of any secure key among the parties. As a result, $I(A_1\ldots A_N\downarrow e)_{\rho^{\mathcal{S}}}=0$ for any separable state. Here, we do not minimize the conditional mutual information over all possible maps of Eve but only consider the specific C.C. attack by Eve. Consequently, we obtain the upper bound to the key, denoted as $I(A_1\ldots A_N|e,C.C.)$.  Moreover, as the conditional mutual information is convex over $e$, the key rate $R$ is upper-bounded by
\begin{equation}
  R\leq \frac{1}{N-1}\sum_e\sum_{\vec{x}}p_{\vec{x}}p(e|\vec{x})I_{\vec{x}}(A_1\ldots A_N|e,C.C.) \label{cc}
\end{equation}
where $I_{\vec{x}}(A_1\ldots A_N|e,C.C.)$ is the conditional mutual information of all parties given inputs $\vec{x}$. 
We now present the main results of this work.

{\it{Bipartite scenario---}}
For simplicity, we first consider the bipartite scenario, where two parties $A_1,A_2$ want to establish a key among themselves using an EB-QKD protocol. Suppose now that the source prepares a two-qubit isotropic state, 
\begin{equation}\label{Wer}
    \rho_{v}=v \proj{\phi^+}+(1-v)\I/4.
\end{equation}
It is well-known that this state is entangled for $v\geq1/3$ and separable otherwise \cite{Werner1987}. The correlations $\{p^{v}(a_1,a_2|x_1,x_2)\}$ are computed as $p^{v}(a_1,a_2|x_1,x_2)=\Tr(M_{a_1|x_1}\otimes M_{a_2|x_2}\rho_v)$. 

\begin{figure}[t!]
    \centering
    \includegraphics[width=\linewidth]{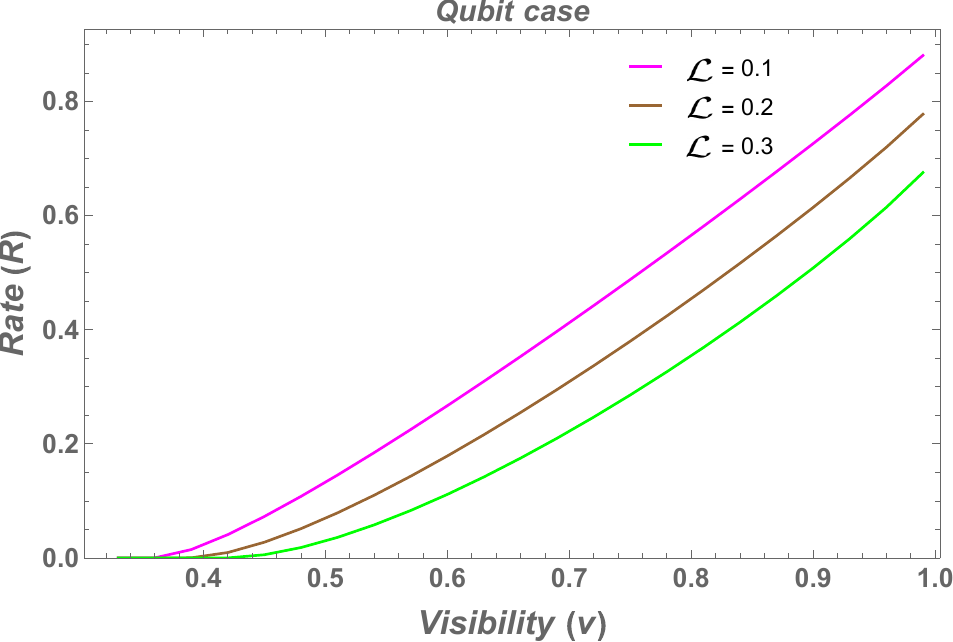}
    \caption{ Under the uniform leakage model, each of the three curves represents the upper bound on the secure key rate achievable in any protocol where both parties publicly announce their measurement inputs (when measuring the isotropic state), for leakage values $\mathcal{L}=0.1,0.2$ and $0.3$.
}
    \label{fig1}
\end{figure}

\begin{thm}\label{theo11}
    Consider the QKD protocol described above with a leakage paramter $\mathcal{L}$, which is implemented using the isotropic state $\rho_v$ \eqref{Wer}. Then, for any two-outcome measurements with all parties, Eve can find a C.C. attack such that no secure key can be extracted for (i) $ v\leq\frac{1}{3}+\frac{2\mathcal{L}}{3-2\mathcal{L}}$ for uniform leakage, (ii) $ v\leq\frac{1}{3}+\frac{2\mathcal{L}}{3(\mathcal{L}+3)}$ for leakage only from "junk".
\end{thm}

We provide a brief sketch of the proof, with full details deferred to Theorem~\ref{theo1} in the appendix. The isotropic state $\rho_v$ can be expressed as a convex combination of the maximally entangled state and the separable isotropic state $\rho_{1/3}$. The latter can then be expressed as a convex mixture of product states. Accordingly, Eve controls the source and distributes these product states, together with the maximally entangled state to the parties. A flag is  attached to indicate which state has been sent. The local probability that outcomes leak is $\mathcal{L}$. Considering the worst-case scenario, that is, outcomes are leaked in the same rounds for both parties, Eve has complete knowledge of Alice’s and Bob’s measurement outcomes in those rounds. Consequently, she assigns her knowledge as $e=(a_1,a_2)$, corresponding to the outcomes at $A_1$ and $A_2$. In contrast, when the maximally entangled state is sent, Eve assigns $e=?$. 

We now consider two different leakage models: (i) uniform leakage, where the leakage of outcomes is independent of Eve’s flag. In this case, outcomes leak with probability $\mathcal{L}$ in every round, so the fraction of leaked outcomes equals $\mathcal{L}$. (ii) leakage only from ``junk'', where outcomes leak only when Eve sends the separable state, i.e., from rounds that cannot generate secure key anyway. In this case, the fraction of leaked outcomes is $3\mathcal{L}(1-v)/2$ for any $v$, which is smaller than in the uniform leakage model. We refer to these rounds as "junk", since whenever Eve sends the separable part of the state, no secure key is generated in those rounds.


Coming back to the proof, since no secure key can be generated from product distributions, Eve’s strategy to weaken the overall QKD protocol is to minimise the intrinsic information $I(A_1,A_2|e=?)$. To this end, she reassigns a fraction $\gamma_{a_1,a_2}$ of events, originally arising from the separable state and yielding outcomes $(a_1,a_2)$ to the class $e=?$. This strategy allows Eve to reduce the intrinsic information to zero, i.e., $I(A_1,A_2|e=?)=0$. This occurs for all $v \leq 1/3 + 2\mathcal{L}/(3-2\mathcal{L})$ in the case of uniform leakage, and for $v \leq 1/3 + 2\mathcal{L}/[3(\mathcal{L}+3)]$ when leakage occurs only from junk rounds.

We observe that for any non-zero leakage $\mathcal{L}$, there is a gap between entangled states and key extractable states. Remarkably, we also observe that this effect persists even if the leakage occurs from the part of the state that could not be used to generate key. Consequently, we can conclude that even if there is an arbitarily small leakage in standard QKD protocols, every entangled state can not be used to generate secure key. 

Using a similar strategy, we find the maximal value of the key that can be securely obtained from any two-qubit isotropic state. Unlike the usual approaches, we do not restrict ourselves to any particular protocol. Instead, we derive a general upper bound to the secure key rate that applies to all standard protocols where both parties announce their inputs [see Appendix B for details]. We now utilize this framework to analyze repeater-based QKD setups.


\begin{figure}[t!]
\includegraphics[width=\columnwidth]{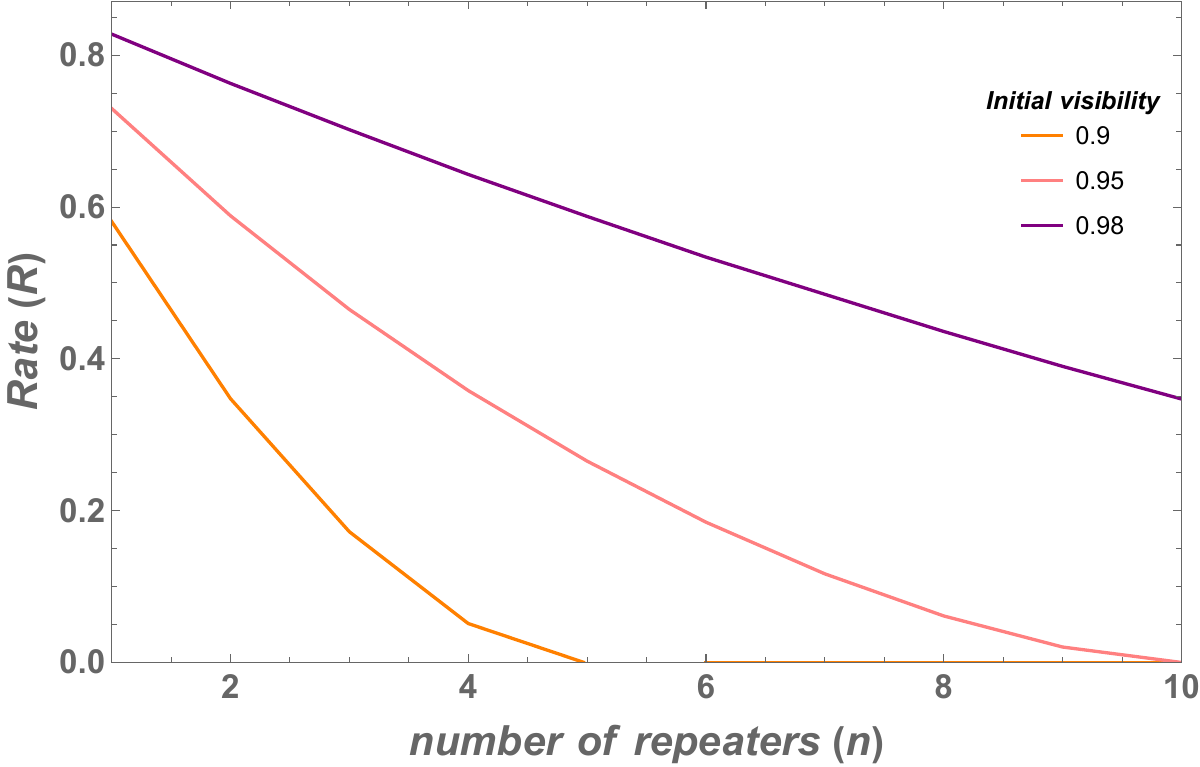}
\caption{ Considering the uniform leakage model, we illustrate the dependence of the secure key rate $R$
on the number of repeaters $n$ for various initial visibilities for a particular leakage $\mathcal{L}=0.1$, showing the rapid degradation of the rate with increasing number of repeaters.} \label{fig2}
\end{figure}

{\it{Repeater-based QKD---}} A quantum repeater enables QKD over very long distances by addressing a core limitation of quantum communication: photon loss and noise in optical fibers. Unlike classical signals, quantum states cannot be amplified due to the no-cloning theorem, which makes direct long-distance QKD nearly impossible beyond a few hundred kilometers. Quantum repeaters solve this problem by breaking a long channel into several shorter elementary links, creating entanglement over each short segment, and then extending that entanglement step by step using Bell-state measurements (BSMs). In this way, rather than trying to send a fragile quantum state across the full distance, the repeater constructs long-range entanglement in a modular fashion. Once the endpoints share the entanglement, they can run a standard EB-QKD protocol (e.g., E–91) to generate a secure key \cite{repeater, rep1,rep2,rep3,rep4}. 

The simplest configuration with a single repeater between Alice and Bob consists of two sources generating entangled states. The first source prepares the state $\rho_{A R_1}$, sending one subsystem to Alice and the other to the repeater, while the other source generates another entangled state $\rho_{R_1 B}$, sending one subsystem to the repeater and the other to Bob. 
In realistic channels, both sources distribute imperfect entangled states, typically modelled as isotropic states with visibility $v$ \eqref{Wer}. At this point, Alice and Bob are not entangled with each other; both are only connected to the repeater. 

To extend entanglement across the full distance, the repeater performs a Bell-state measurement (BSM) on the subsystems $R_1^{(1)}$ and $R_1^{(2)}$. The measurement swaps the entanglement from the measurement to Alice and Bob, up to a local Pauli correction that depends on the BSM outcome and can be communicated classically. Crucially, however, when the two initial links have visibility $v$, the shared state that results after swapping is another isotropic state but with reduced visibility $v^2$. 
Consequently, a chain of $n$ repeaters yields a final visibility $v^{2n}$ on the state shared by Alice and Bob. Consequently, we conclude from Theorem \ref{theo1} that one cannot obtain a secure key for $v^{2n}\leq1/3+2\mathcal{L}/(3-2\mathcal{L})$ (considering a uniform leakage model). Therefore, for a given visibility $v$, the maximum number of repeater nodes $n_{\max}$ compatible with secure key
\begin{eqnarray}
    n_{\max}\leq \frac{\log_2(1/3+2\mathcal{L}/(3-2\mathcal{L}))}{2\log_2v}.
\end{eqnarray}
As is clear from this expression, in an idealised error-free scenario, any standard protocol would, in principle, allow an arbitrarily large number of repeater nodes.

Strikingly, introducing even a modest imperfection in state preparation, on the order of only five percent (corresponding to $v=.95$) and $\mathcal{L}=.1$, drastically limits this scalability, restricting the generation of a secure key to at most ten repeaters. This highlights the extreme sensitivity of long-distance repeater-based QKD to small preparation errors. Note that in the device-independent setting, the maximal number of repeaters when restricting to the CHSH protocol was found in \cite{sadhu2024}.

In Fig. \ref{fig2}, we illustrate this behaviour, as the key rate decreases rapidly as the number of repeaters increases. As discussed earlier, after \(n\) repeaters an initial isotropic-state visibility \(v\) is reduced to an effective visibility \(v^{2n}\). For each value of \(v^{2n}\), the secure key rate is computed by 
optimising the C.C. strategy (capturing the worst-case eavesdropping scenario), and then subsequently maximizing over Alice’s and Bob’s measurement settings to determine the highest achievable secure key rate consistent with that attack. Repeating this procedure for initial visibilities \(v = 0.9, 0.95,\) and \(0.98\) with $\mathcal{L}=.1$ yields the curves displayed in Fig.~\ref{fig2}.


{\it{Generalisation to arbitrary dimension and number of parties---}}
We now consider the most general QKD scenario as discussed above, where arbitrary number of parties want to establish a key among themselves using arbitrary-dimensional systems. Recall that we consider protocols where the parties announce their inputs. To this end, suppose that the source prepares a $N$-qudit isotropic state 
\begin{equation}\label{Iso}
    \rho_{v}=v \proj{\phi_{d,N}}+(1-v)\frac{\I}{d^N}
\end{equation}
where $\ket{\phi_{d,N}}$ is the generalised GHZ state given by
\begin{eqnarray}
    \ket{\phi_{d,N}}=\frac{1}{\sqrt{d}}\sum_{i=0}^{d-1}\ket{i}^{\otimes N}.
\end{eqnarray}
It is well-known that the above state is fully separable for $v\leq1/(1+d^{N-1})$ and entangled across some bipartition otherwise \cite{Dur2000}. The correlations $\{p^{v}(\vec{a}|\vec{x})\}$ are computed as $p^{v}(\vec{a}|\vec{x})=\Tr(\bigotimes_{i=1}^NM_{a_i|x_i}\rho_v)$. 

\begin{figure}[t!]
    \centering
    \includegraphics[width=\linewidth]{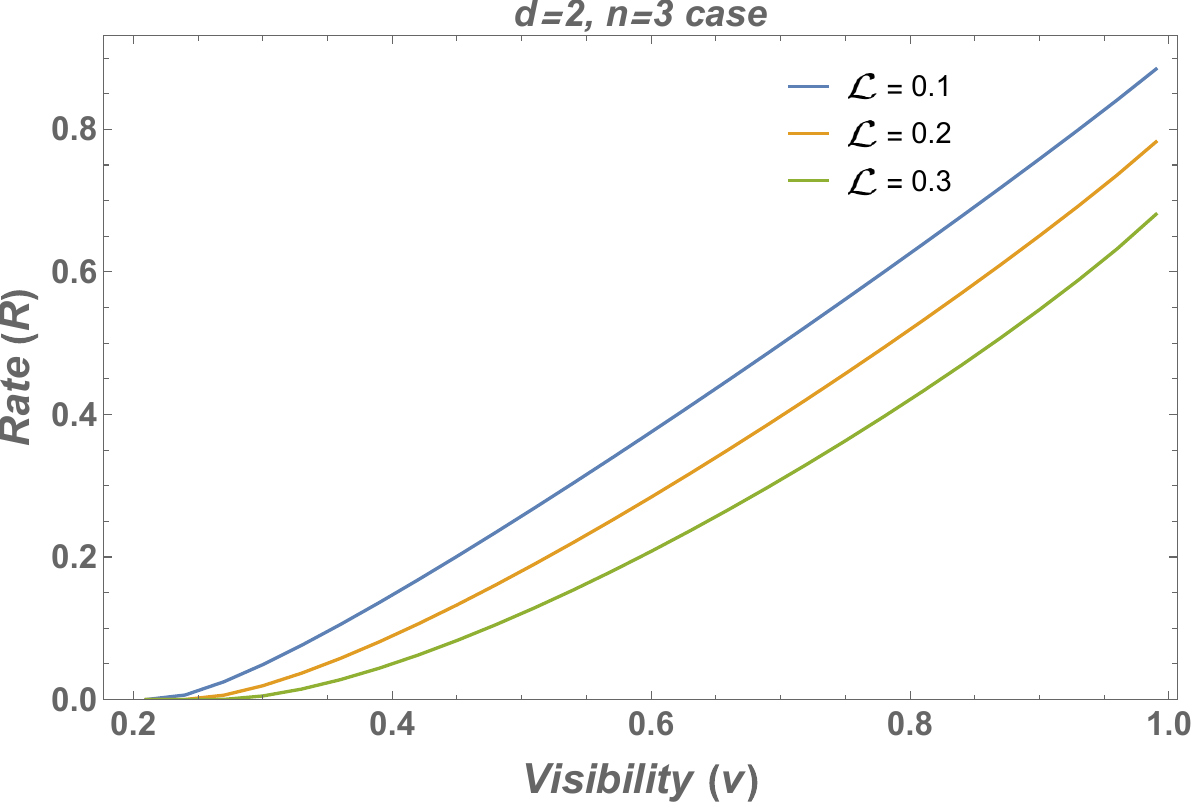}
    \caption{
Under the uniform leakage model, the three curves depict upper bounds on the secure key rate as a function of the visibility \(v\) of the isotropic state \(\rho_v\) for the  \(d = 2, N = 3\) scenario, with each curve corresponding to a specific leakage value.}

    \label{fig3}
\end{figure}

\begin{thm}
    Consider the QKD protocol described above with a leakage parameter $\mathcal{L}$, which is implemented using the isotropic state $\rho_v$ \eqref{Iso}. Then, for any $d-$outcome measurements with all parties, Eve can find a C.C. attack such that no secure key can be extracted for (i) $ v\leq\frac{1}{1+d^{N-1}}+\frac{d^{N-1}\mathcal{L}}{(1+d^{N-1})(1-\mathcal{L})+\mathcal{L}}$ for uniform leakage, (ii) $ v\leq\frac{1}{1+d^{N-1}}+\frac{d^{N-1}\mathcal{L}}{(1+d^{N-1})(\mathcal{L}+1+d^{N-1})}$ for leakage only from "junk".
\end{thm}

The proof of the above theorem is provided in Theorem \ref{theo1} in the appendix. For the special cases $d=3, N=2$ and $d=2, N=3$, we compute the upper bound on the secure key rate achievable with standard QKD protocols for an isotropic state (see Fig.~\ref{fig3}). The computation is based on finding the best CC attack by Eve, given arbitrary projective measurements performed by the parties [see Appendix B for details].

\begin{figure}[t!]
    \centering
    \includegraphics[width=\linewidth]{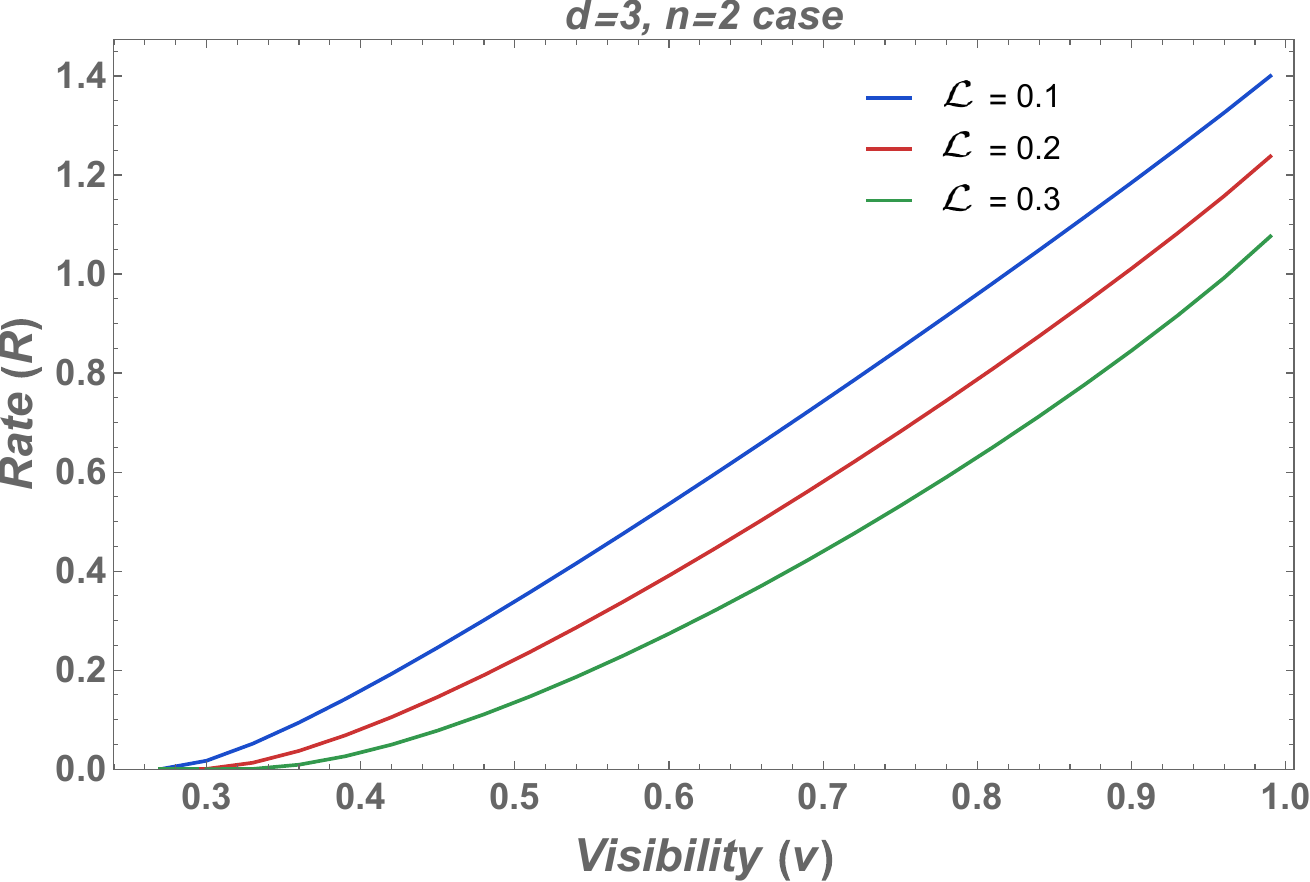}
    \caption{
   Considering the uniform leakage model, we illustrate the upper bounds on the secure key rate as a function of the visibility $v$ of the isotropic state $\rho_v$ for the  \(d = 3, N = 2\) scenario, shown three different leakage values. Here, the secret key rate can exceed one bit, since a three-dimensional quantum system can, in principle, generate a three-level key symbol per round, allowing up to \(\log_2 3\) bits of secret key.}

    \label{fig4}
\end{figure}


{\it{Discussions---}} Our findings highlight a fundamental limitation in the practical use of entanglement for QKD. While entanglement is widely regarded as the essential resource enabling quantum security, our results demonstrate that it is not, by itself, sufficient for secure key generation in standard entanglement-based practical QKD protocols. In particular, we have shown that when the input choices are publicly announced, as in the E91 and entanglement-based BB84 schemes, even an arbitrarily small amount of leakage of classical side information implies that there exist entangled states that fail to yield any secret key. Counter-intuitively, we find that this gap persists even if the leakage happens only from "junk" rounds. This means that even if leakage occurs only in rounds from which no secure key can be generated, it is still sufficient to reveal a gap between entangled states and key-generating states.

To do this, we identified a class of two-qubit isotropic states that are entangled yet useless for key distillation under most protocols, thereby underscoring the practical distinction between entanglement and key-generating capability. Extending this observation, we showed that a similar limitation arises in higher-dimensional and multipartite systems, where certain isotropic states are entangled but cannot be used to extract a secure key. These results suggest that the mere presence of entanglement does not guarantee practical advantage for cryptographic tasks. They further point towards the need to identify operational properties beyond entanglement that determine whether a state can support secure key generation in realistic settings.
Finally, we demonstrated that even a small noise in the source drastically reduces the scalability of repeater-based QKD architectures. 


Several questions arise from our findings. A key direction is to develop a general criterion distinguishing entangled states that can yield a secret key from those that cannot under realistic conditions. Extending these ideas to other settings, such as measurement-device-independent or prepare-and-measure protocols, is important for understanding the fundamental resources enabling secure key generation. In particular, applying the C.C. attack in prepare-and-measure scenarios would be interesting, since the source is held by one of the parties and cannot be accessed or manipulated by the adversary, which changes the structure of possible attacks. In this work, we considered one of the simplest possible attacks; however, under more sophisticated attacks, an even larger set of entangled states could be shown to be useless for secure key generation.


{\textit{Acknowledgements---}}
This project was funded within the National Science Centre, Poland, grant Opus 25, UMO-2023/49/B/ST2/02468.

%

\onecolumngrid
\appendix
\section{Appendix A: Proof of the main theorem}
\setcounter{thm}{0}
\begin{thm}\label{theo1}
    Consider the QKD protocol described above with a leakage parameter $\mathcal{L}$, which is implemented using the isotropic state $\rho_v$ \eqref{Iso}. Then, for any $d-$outcome measurements with all parties, Eve can find a C.C. attack such that no secure key can be extracted for (i) $ v\leq\frac{1}{1+d^{N-1}}+\frac{d^{N-1}\mathcal{L}}{(1+d^{N-1})(1-\mathcal{L})+\mathcal{L}}$ for uniform leakage, (ii) $ v\leq\frac{1}{1+d^{N-1}}+\frac{d^{N-1}\mathcal{L}}{(1+d^{N-1})(\mathcal{L}+1+d^{N-1})}$ for leakage only from "junk".
\end{thm}

\begin{proof}

Let us construct the convex combination attack. For this purpose, the state generated by the source is mixed with a fully separable and entangled state as
\begin{equation}
\rho_{v}=\frac{q_v}{(1+d^{N-1})}\left(\proj{\phi_{d,N}}+\frac{1}{d}\I\right)+(1-q_v)\proj{\phi_{d,N}}
\end{equation}
which implies that $q_v=\frac{(1+d^{N-1})(1-v)}{d^{N-1}}$. It is thus straightforward to observe that the statistics with all parties remain invariant under this mixing. The fully separable isotropic state can further be represented by a convex combination of product states as
\begin{eqnarray}
\rho^{\mathrm{sep}}=\frac{1}{(1+d^{N-1})}\left(\proj{\phi_{d,N}}+\frac{1}{d}\I\right)=\sum_{i}p_i\bigotimes_{j=1}^N \proj{\mu^{(j)}_i}
\end{eqnarray}
where $\ket{\mu^{(j)}_i}\in \mathbb{C}^d$. Consequently, Eve sends each product state with a flag $f=f_{i}$ and the entangled state with a flag $f=?$ thus the joint state is
\begin{eqnarray}
    \rho^v_{A_1\ldots A_NE}=q_v\left(\sum_{i}p_i\bigotimes_{j=1}^N \proj{\mu^{(j)}_i}\otimes\proj{f_{i}}_E\right)+(1-q_v)\proj{\phi^+}\otimes\proj{?}_E.
\end{eqnarray}

Now, in any practical scenario, there is always some leakage of side information to the adversary. Here, we suppose that for any party, the outcomes in the key generating rounds are leaked with a probability $\mathcal{L}$.
Let us first consider the uniform leakage model, that is, leakage happens independently of the flag with the same probability $\mathcal{L}$.
Equivalently, one can understand $\mathcal{L}$ as the fraction of the number of outcomes leaked out versus the total number of outcomes for each party in the key generation rounds. In the worst-case scenario, Eve learns the outcomes of all parties in the same round, and thus, $\mathcal{L}$ fraction of the time, the outcomes of all parties are known to Eve when any state is sent. Consequently, the joint probability distribution of all parties along with Eve is given by
\begin{eqnarray}\label{probafl}
    p^v(\vec{a},e|\vec{x})&=&q_v(1-\mathcal{L})\left(\sum_{i}p_ip^{i}(\vec{a}|\vec{x})\delta_{e,f_{i}}\right)+q_v\mathcal{L}p^{\mathrm{sep}}(\vec{a}|\vec{x})\delta_{e,(\vec{a})}\nonumber\\&+&(1-q_v)(1-\mathcal{L})p^{\mathrm{ent}}(\vec{a}|\vec{x})\delta_{e,?}+(1-q_v)\mathcal{L}p^{\mathrm{ent}}(\vec{a}|\vec{x})\delta_{e,(\vec{a})}.
\end{eqnarray}
Here $p^{i}(\vec{a}|\vec{x})$ are product distributions as they are generated using the states $\bigotimes_{j=1}^N \proj{\mu^{(j)}_i}$ and $p^{\mathrm{ent}}(\vec{a}|\vec{x}),p^{\mathrm{sep}}(\vec{a}|\vec{x})$ is the probability distribution when the underlying state is maximally entangled and the separable isotropic state respectively.

Let us introduce a mixing parameter $\gamma_{\vec{a},\vec{x}}$, based on which the outcomes from the leaked part of $p^{\mathrm{sep}}(\vec{a}|\vec{x})$ are mixed with $?$. After mixing, the new probability distribution of all parties and Eve is given by
\begin{eqnarray}\label{eqA12}
 p(\vec{a},e|\vec{x})=q_v(1-\mathcal{L})\left(\sum_{i}p_ip^{i}(\vec{a}|\vec{x})\delta_{e,f_{i}}\right)+q_v\mathcal{L}(1-\gamma_{\vec{a},\vec{x}})p^{\mathrm{sep}}(\vec{a}|\vec{x})\delta_{e,\vec{a}}+(1-q_v)\mathcal{L}p^{\mathrm{ent}}(\vec{a}|\vec{x})\delta_{e,\vec{a}}\nonumber\\+q_v\mathcal{L}\gamma_{\vec{a},\vec{x}}p^{\mathrm{sep}}(\vec{a}|\vec{x})\delta_{e,?}+(1-q_v)(1-\mathcal{L})p^{\mathrm{ent}}(\vec{a}|\vec{x})\delta_{e,?}.\qquad
\end{eqnarray}

As described in the manuscript, the key rate is upper-bounded as
\begin{equation}
  R\leq \frac{1}{N-1}\sum_e\sum_{\vec{x}}p_{\vec{x}}p(e|\vec{x})I_{\vec{x}}(A_1\ldots A_N|e,C.C.) 
\end{equation}
Let us observe that $I_{\vec{x}}(A_1\ldots A_N|e=f_i,C.C.)=0$ for all $i$, as $p^{i}(\vec{a}|\vec{x})$ are product distributions along with  $I_{\vec{x}}(A_1\ldots A_N|e=\vec{a},C.C.)=0$ as the outcomes are known to Eve. Consequently, we can immediately conclude that key rate is upper bounded as
\begin{eqnarray}
   R\leq \frac{1}{N-1} \sum_{\vec{x}}p_{\vec{x}}p(e=?|\vec{x})I_{\vec{x}}(A_1\ldots A_N|e=?,C.C.).
\end{eqnarray}
We now find conditions when $I_{\vec{x}}(A_1\ldots A_N|e=?,C.C.)=0$, that is, we find $q_v$ for which there is a well-defined mixing parameter $\gamma_{\vec{a},\vec{x}}$ such that $0\leq\gamma_{\vec{a},\vec{x}}\leq1$ for any $\vec{a},\vec{x}$. For this purpose, we first observe that for any $\vec{x}$ we can always find a set of outputs $\vec{a'}$ such that $p^{\mathrm{ent}}(\vec{a'}|\vec{x})=\max_{\vec{a}}\{p^{\mathrm{ent}}(\vec{a}|\vec{x})\}$. Evaluating the conditional probability distribution, $p^{v}(\vec{a}|e=?,\vec{x})$ using the Bayes rule $p^{v}(\vec{a}|e=?,\vec{x})=p(\vec{a},e=?|\vec{x})/p(e=?|\vec{x})$ to obtain that
\begin{eqnarray}\label{eq15}
    p^{v}(\vec{a}|e=?,\vec{x})=\frac{1}{p(e=?|\vec{x})}\left(q_v\mathcal{L}\gamma_{\vec{a},\vec{x}}p^{\mathrm{sep}}(\vec{a}|\vec{x})+(1-q_v)(1-\mathcal{L})p^{\mathrm{ent}}(\vec{a}|\vec{x}) \right).
\end{eqnarray}
The simplest criteria for $I_{\vec{x}}(A_1,\ldots,A_N|e=?)=0$ is that all the conditional probabilities $ p^{v}(\vec{a}|e=?,\vec{x})$ are equal to each other, that is, the mixing parameters are adjusted such that all the conditional probabilities in Eq. \eqref{eq15} for each $\vec{a},\vec{x}$ are equal. Consequently, we get the condition for all $\vec{a},\vec{x}$
\begin{eqnarray}\label{eq17}
    (1 -  q_v)(1-\mathcal{L}) p^{\mathrm{ent}}(\vec{a'}|\vec{x}) = q_v\mathcal{L}\gamma_{\vec{a},\vec{x}}p^{\mathrm{sep}}(\vec{a}|\vec{x})+(1-q_v)(1-\mathcal{L})p^{\mathrm{ent}}(\vec{a}|\vec{x})\nonumber\\
    \implies \gamma_{\vec{a},\vec{x}}=\frac{(1-\mathcal{L})(1-q_v)(p^{\mathrm{ent}}(\vec{a'}|\vec{x})-p^{\mathrm{ent}}(\vec{a}|\vec{x}))}{q_v\mathcal{L}p^{\mathrm{sep}}(\vec{a}|\vec{x}) }.
\end{eqnarray}
The problem now simplifies to finding a $\gamma_{\vec{a},\vec{x}} \in [0,1]$ . If there exists such a $\gamma_{\vec{a},\vec{x}}$, then Eve can use the above described strategy to guess the key. Since it is clear from the above expression that $\gamma_{\vec{a},\vec{x}}\geq0$, we need to only establish conditions under which $\gamma_{\vec{a},\vec{x}}\leq1$.

Let us now recall that for $d-$outcome projective measurements, we have that $p^{\mathrm{ent}}(\vec{a'}|\vec{x})\leq 1/d$ and  $p^{\mathrm{sep}}(\vec{a'}|\vec{x})=\frac{1}{1+d^{N-1}}(p^{\mathrm{ent}}(\vec{a}|\vec{x})+1/d)$. Thus, $\gamma_{\vec{a},\vec{x}}$from \eqref{eq17} is upper-bounded as follows
\begin{eqnarray}
    \gamma_{\vec{a},\vec{x}}=\frac{(1+d^{N-1})(1-\mathcal{L})(1-q_v)(p^{\mathrm{ent}}(\vec{a'}|\vec{x})-p^{\mathrm{ent}}(\vec{a}|\vec{x}))}{q_v\mathcal{L}(p^{\mathrm{ent}}(\vec{a}|\vec{x})+1/d) }\leq \frac{(1+d^{N-1})(1-\mathcal{L})(1-q_v)(1/d)}{q_v\mathcal{L}/d }.
\end{eqnarray}
As the above quantity needs to be less than or equal to $1$, we find that
\begin{eqnarray}\label{eqA18}
     q_v\geq\frac{(1+d^{N-1})(1-\mathcal{L})}{(1+d^{N-1})(1-\mathcal{L})+\mathcal{L}}
\end{eqnarray}
which on recalling that $q_v=\frac{(1+d^{N-1})(1-v)}{d^{N-1}}$ gives us
\begin{eqnarray}
    v\leq 1-\frac{d^{N-1}(1-\mathcal{L})}{(1+d^{N-1})(1-\mathcal{L})+\mathcal{L}}=\frac{1}{1+d^{N-1}}+\frac{d^{N-1}\mathcal{L}}{(1+d^{N-1})(1-\mathcal{L})+\mathcal{L}}.
\end{eqnarray}
Consequently, for any $\mathcal{L}\geq 0$, Eve can obtain a mixing parameter $\gamma_{\vec{a},\vec{x}}$ such that the intrinsic information is $0$ even from an entangled state and thus no secure key can be extracted.

Let us consider now that the leakage only happens from "junk" rounds, that is, from rounds when separable state is sent by Eve. Thus, the joint probability distribution of all parties and Eve after the leakage is almost the same as \eqref{probafl}, but since there is no leakage when the flag is $?$, we get
\begin{eqnarray}
    p^v(\vec{a},e|\vec{x})&=&q_v(1-\mathcal{L})\left(\sum_{i}p_ip^{i}(\vec{a}|\vec{x})\delta_{e,f_{i}}\right)+q_v\mathcal{L}p^{\mathrm{sep}}(\vec{a}|\vec{x})\delta_{e,(\vec{a})}+(1-q_v)p^{\mathrm{ent}}(\vec{a}|\vec{x})\delta_{e,?}.
\end{eqnarray}
Consequently, the fraction of outcomes leaked is $\mathcal{L}q_v$, which is lower than the former case. Following the same steps as above from Eq. \eqref{eqA12} to \eqref{eqA18}, we obtain
\begin{eqnarray}
     q_v\geq\frac{1+d^{N-1}}{1+d^{N-1}+\mathcal{L}}
\end{eqnarray}
which on recalling that $q_v=\frac{(1+d^{N-1})(1-v)}{d^{N-1}}$ gives us
\begin{eqnarray}
    v\leq \frac{1}{1+d^{N-1}}+\frac{d^{N-1}\mathcal{L}}{(1+d^{N-1})(\mathcal{L}+1+d^{N-1})}.
\end{eqnarray}

Consider now general $d-$outcome measurements acting on $d-$dimensional Hilbert space. As shown in \cite{D_Ariano_2005}, any POVM can be written as a convex combination of extremal POVM's. Some examples of such extremal POVMs in arbitrary dimensions and number of outcomes can be found in \cite{Sarkar_2023}. For $d-$outcome measurements, the extremal POVMs are projective measurements. Thus, any $d-$outcome POVM $\{N_a\}$ can be expressed as $\sum_ip_i\{M_a^i\}$ where $M_a^i$ are projectors. As intrinsic information \eqref{intinf11} with general measurements can be expressed as a convex combination of intrinsic information via projective measurements if Eve knows the mixing labels $i$, that is,
\begin{eqnarray}
    I_{\vec{x}}(A_1\ldots A_N\downarrow e)_{\{N_a\}}\leq \sum_i p_i I_{\vec{x}}(A_1\ldots A_N\downarrow e)_{\{M_{a}^i\}}.
\end{eqnarray}
Now, Eve can straightaway know the labels $i$ as she can attach a flag on the states which activate particular label at the parties.
Consequently, we can straightaway conclude that if the secure key obtainable using $d-$outcome projective measurements is $0$, then $d-$outcome general measurements do not help. Moreover, one can also understand this by noting that, with POVMs, the correlations between Alice and Bob decrease, and thus the optimal strategy for finding the key is via projective measurements.
This completes the proof.

\section{Appendix B: Maximum obtainable rate for any isotropic state using projective measurements}

Let us recall from Eq. \eqref{cc} of the manuscript that the maximal achievable rate when Eve performs a C.C. attack is given by,
\begin{equation}
  R\leq \frac{1}{N-1}\sum_e\sum_{\vec{x}}p_{\vec{x}}p(e|\vec{x})I_{\vec{x}}(A_1\ldots A_N|e,C.C.) \label{upb_gen}  
\end{equation}
where $I_{\vec{x}}(A_1\ldots A_N|e,C.C.)$ is the conditional mutual information of all parties given inputs $\vec{x}$.

Eve's strategy to execute the C.C. attack is explicitly given in the proof of theorem \ref{theo1}. Using this as stated above in Eq. \eqref{upb_gen}, the upper bound to the key rate is given by
\begin{eqnarray}
    R\leq \frac{1}{N-1}\sum_{\vec{x}}p_{\vec{x}}p(e=?|\vec{x})I_{\vec{x}}(A_1,A_2,\ldots,A_N|e=?)\leq \frac{1}{N-1}\max_{\vec{x}}I_{\vec{x}}(A_1,A_2,\ldots,A_N|e=?).
\end{eqnarray}

In Fig. 1 of the manuscript, we find the maximal rate for the case when $d=2, N=2$. To do this, we first express the conditional mutual information $I_{\vec{x}}(A_1,A_2|e=?)$ as
\begin{eqnarray}
    I_{\vec{x}}(A_1,A_2|e=?)=H(A_1|e=?)+H(A_2|e=?)-H(A_1,A_2|e=?)
\end{eqnarray}
where $H(.)$ is the Shannon entropy. The Shannon entropy can be expressed in terms of the conditional joint probability distribution $p^v(\vec{a}|e=?,\vec{x})$ as
\begin{eqnarray}\label{I_d=2}
    I_{x_1,x_2}^v(A_1,A_2|e=?)=\sum_{a_1,a_2=0,1}p^v(e=?)p^v(a_1,a_2|e=?,x_1,x_2)\log \frac{p^v(a_1,a_2|e=?,x_1,x_2)}{p^v(a_1|e=?,x_1)p^v(a_2|e=?,x_2)}
\end{eqnarray}
where the superscript $v$ signifies that the quantities are computed for isotropic state \eqref{Iso} for $d=2, N=2$ with a visibility $v$.
Referring back to the proof of theorem \ref{theo1}, we observe that Eve's best strategy relies on finding the mixing parameters $\gamma_{\vec{a},\vec{x}}$ such that the conditional mutual information is minimised. Consequently, expressing the conditional probability distribution $p^v(a_1,a_2|e=?,x_1,x_2)$ as done in Eq. \eqref{eq15} and then using the fact that $A_1,A_2$ correlations are no-signalling, we obtain from \eqref{I_d=2} that $I_{x_1,x_2}^v(A_1,A_2|e=?)$ is a function of the mixing parameters $\gamma_{\vec{a},\vec{x}}$ such that $0\leq\gamma_{\vec{a},\vec{x}}\leq1$. Thus, the maximal rate for any visibility $v$ is given by
\begin{eqnarray}
R\leq \max_{\vec{x}}\min_{\gamma_{\vec{a},\vec{x}}}I^v_{\vec{x}}(A_1,A_2|e=?)(\gamma_{\vec{a},\vec{x}}).
\end{eqnarray}

One can straightaway generalise the above statement for any $d,N$ for any isotropic state of visibility $v$ to obtain
\begin{eqnarray}
R\leq \frac{1}{N-1}\max_{\vec{x}}\min_{\gamma_{\vec{a},\vec{x}}}I^v_{\vec{x}}(A_1,\ldots,A_N|e=?)(\gamma_{\vec{a},\vec{x}}).
\end{eqnarray}
Notice that for higher $N$ one has to use the chain rule to obtain  $I^v_{\vec{x}}(A_1,\ldots,A_N|e=?)$ from \eqref{I_d=2} which is technically involving.

Below, we describe the numerical procedure used to obtain the key-rate plots shown in Figs.~1, 3, and 4. 
All figures are computed under the uniform leakage model. Each figure contains three curves, 
corresponding to different leakage values, with each curve obtained using the same numerical optimization 
procedure described below.
We consider three scenarios, namely (i) $d=2, N=2$, (ii) $d=2, N=3$, and (iii) $d=3, N=2$. In all cases, the secret key rate is evaluated as a function of the visibility using the same two-step max--min optimization strategy: we first minimize over Eve’s parameters $\gamma_{\vec{a},\vec{x}}$ and then maximize over the honest parties’ measurement settings. The resulting curves therefore capture both 
the optimal honest-party strategy and the worst-case eavesdropping attack. For the case $d=2, N=2$, shown in Fig.~1, we consider the bipartite qubit scenario. The horizontal axis 
represents the visibility, while the vertical axis shows the optimized key rate obtained through the above optimization procedure.
For $d=2, N=3$, shown in Fig.~3, we consider the tripartite qubit scenario. Alice performs measurements 
in the $Z$ basis, while Bob and Charlie use projective measurements in the $X$--$Z$ plane parametrized 
by angles $\theta_1$ and $\theta_2$. In the numerical optimization, without loss of generality, Alice’s and Bob’s measurements are 
fixed to the $Z$ basis, and the maximization is performed only over $\theta_2$. The optimized key rate 
is then plotted as a function of the visibility.

Finally, for the bipartite qutrit case $(d=3, N=2)$, shown in Fig.~4, both Alice and Bob perform 
projective measurements in the qutrit $Z$ basis. Applying the same max--min optimization procedure 
described above, we compute the key rate as a function of the visibility.

\end{proof}
\end{document}